\documentclass[a4paper,11pt]{scrartcl}

\usepackage{amssymb}
\usepackage{amsmath,amsthm,graphicx}
\usepackage{xspace}
\usepackage[hidelinks]{hyperref}
\usepackage[square,numbers]{natbib}
\usepackage[margin=1in]{geometry}

\newtheorem{definition}{Definition}
\newtheorem{lemma}{Lemma}
\newtheorem{theorem}{Theorem}

\graphicspath{{figures/}}

\title{Switches in Eulerian graphs}

\author{Ahad N. Zehmakan\footnote{ETH Zurich, Department of Computer Science, abdolahad.noori@inf.ethz.ch}, Jerri Nummenpalo\footnote{ETH Zurich, Department of Computer Science, njerri@inf.ethz.ch}, Alexander Pilz\footnote{Graz University of Technology, Institute of Software Technology, apilz@ist.tugraz.at},\\ Daniel Wolleb-Graf\footnote{ETH Zurich, Department of Computer Science, daniel.graf@inf.ethz.ch}}

\date{}

\begin{document}

\maketitle

% \author{Daniel Graf\corref{cor1}}%\fnref{label2}}
% \ead{daniel.graf@inf.ethz.ch}
% %% \ead[url]{home page}
% %\fntext[label2]{Test2}
% \cortext[cor1]{ETH Z\"urich, Department of Computer Science, 8092 Z\"urich, Switzerland.}
% %% \address{Address\fnref{label3}}
% %% \fntext[label3]{}
% \author{Jerri Nummenpalo\corref{cor1}}%\fnref{label2}}
% \ead{njerri@inf.ethz.ch}
% \author{Alexander Pilz\corref{cor2}}%\fnref{label2}}
% \ead{apilz@ist.tugraz.at}
% \cortext[cor2]{Graz University of Technology, Institute of Software Technology, 8010 Graz, Austria.}
% \author{Ahad N. Zehmakan\corref{cor1}}%\fnref{label2}}
% \ead{abdolahad.noori@inf.ethz.ch}

\begin{abstract}
We show that the graph transformation problem of turning a simple graph into an Eulerian one by a minimum number of single edge switches is NP-hard.
Further, we show that any simple Eulerian graph can be transformed into any other such graph by a sequence of 2-switches (i.e., exchange of two edge pairs), such that every intermediate graph is also Eulerian.
However, finding the shortest such sequence also turns out to be an NP-hard problem.
\end{abstract}


\section{Introduction}\label{sec:intro}
In an \emph{edge modification problem}, we are given a graph with the goal of obtaining another graph that is of a certain class by modifying its edge set (i.e., by applying a small number of edge removals or edge additions).
Many such problems have been addressed in algorithmic graph theory (see, e.g.,~\cite[Table~1]{edge_modification_problems}), and in particular, the problem of making a graph Eulerian is already well-studied.
In this work, we contribute further results to this area;
we consider problems of modifying graphs by \emph{$k$-switches}.

\begin{definition}
For a positive integer $k$ a \emph{$k$-switch} on a graph $G = (V,E)$ is the operation of removing $k$ edges from $E$ and adding $k$ edges to obtain a simple graph $G' = (V, E')$ with $|E'| = |E|$.
\end{definition}

Adding another variant to previous results, we show that finding the minimum number of $1$-switches to make a simple graph Eulerian is NP-hard.
Recall that a graph is Eulerian if it is connected and if every vertex has even degree. 
We then turn to the problem of modifying Eulerian graphs by $2$-switches;
we show that for any two Eulerian graphs on $n$ vertices with the same number of edges, one can be transformed into another by a sequence of $2$-switches with the invariant that each intermediate graph is also Eulerian.

\paragraph{Related work.}
An overview of edge modification problems to obtain certain graphs is given in~\cite{edge_modification_problems}.
As can be seen there, many of these problems are NP-complete.
However, for Eulerian graphs, there are some surprising results.
The well-known Chinese Postman Problem (shown to be in P by Edmonds and Johnson~\cite{edmonds_johnson}) can be considered as a way of augmenting a graph to an Eulerian multigraph by duplicating as few edges as possible.
We focus on undirected simple graphs (some of the references also cover digraphs and multigraphs). Dorn et al.~\cite{eulerian_extension} show that the problem of adding the minimum number of edges to make a given graph Eulerian can be done in polynomial time.
The problem of removing the minimum number of edges, however, is NP-hard but in FPT~\cite{eulerian_deletion}.
Dabrowski et al.~\cite{eulerian_editing} recently considered the variant that allows both adding and removing edges;
they show that the problem of determining the minimum number of required operations is tractable.
We consider the variant in which the numbers of added and removed edges have to match and show that the problem becomes NP-hard again.

\begin{theorem}\label{thm:eulerian_flip_NP_complete}
Given a graph $G$ and an integer $\ell$, the problem of deciding whether $G$ can be turned into an Eulerian graph using at most $\ell$ $1$-switches is NP-complete.
\end{theorem}

%This modification can be interpreted as iteratively replacing one edge by another.
%Such an operation is known as \emph{flip} or \emph{switch}.
%Such $k$-switches are used throughout the literature~\cite{missing}.
%There, i
It is usually required that after a $k$-switch, the resulting graph belongs to a prescribed class (e.g., maximal planar graphs~\cite{flips_planar_graphs}, a graph with a certain degree sequence~\cite{taylor}, etc.).
One may therefore ask whether an Eulerian graph can be transformed into any other by a sequence of $k$-switches, such that every intermediate graph is also Eulerian.
Clearly, this problem does not make sense for $1$-switches.
We answer this question in the affirmative for $k=2$ which of course implies the result for all $k\geq 2$.
For the special case of cycles, it follows from a result by \citet{solomon2003sorting} that finding the shortest such sequence is NP-hard.

\begin{theorem}\label{thm:connected}
\sloppy
Let $G = (V, E)$ and $H = (V, F)$ be any two labeled Eulerian graphs with $m = |E| = |F|$ edges.
Then there is a sequence $(G = G_0, G_1, \dots , G_\ell = H)$ of labeled graphs such that each graph $G_i$ is Eulerian and $G_{i+1}$ can be obtained from $G_i$ by a 2-switch.
\end{theorem}

We rely on a similar result regarding graphs with a given degree sequence by Taylor~\cite{taylor}.
For a recent account on 2-switches in that setting see, e.g.,~\cite{barrus}.

\section{1-switches to Eulerian graphs}
\begin{proof}[Proof of Theorem~\ref{thm:eulerian_flip_NP_complete}]
The problem is clearly in NP.
We use an approach similar to the one of Cygan et al.~\cite{eulerian_deletion}.
%The problem is in NP since for a graph with $m \geq 2$ edges there is always an Eulerian flip sequence of length at most $m - 1$.
%We show NP-hardness by reducing from the problem of determining whether a 3-regular graph is Hamiltonian, which is known to be NP-complete~\cite{garey1976some}.
We reduce from the NP-complete problem of determining whether a 3-regular graph is Hamiltonian~\cite{garey1976some}.
Let $G$ be a 3-regular graph with $n$ vertices and let $v$ be a vertex of $G$.
Note that $n$ is even.
Form a new graph $G'$ from $G$ by attaching $n$ paths of length 2 to the vertex $v$ (see \figurename~\ref{fig:hardness example}).
\begin{figure}[t]
\centering
\includegraphics{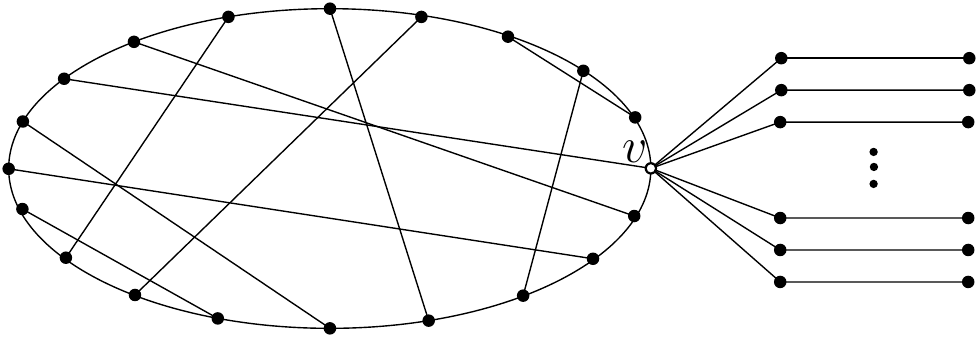}
\caption{An example for the construction in the proof of Theorem~\ref{thm:eulerian_flip_NP_complete} where $n = 20$.
}
\label{fig:hardness example}
\end{figure}
We claim that $G'$ can be made Eulerian using $n/2$ $1$-switches if and only if $G$ is Hamiltonian.
%Clearly $n/2$ switches are needed to turn $G'$ Eulerian since it has $2n$ odd degree vertices and with a single switch at most 4 vertices may change the parity of their degree.
In the latter case, $G$ consists of a Hamiltonian cycle and a perfect matching. 
By switching the edges of the matching to the endpoints of the paths, we can make $G$ Eulerian with $n/2$ switches.
%, since the resulting graph is connected and every vertex has even degree.

Assume now that $n/2$ switches are sufficient to make $G'$ Eulerian.
As $G'$ has $2n$ vertices of odd degree, every switch has to remove an edge that connects two odd-degree vertices and add an edge between two odd-degree vertices.
Therefore, the edges in the attached paths in $G'$ must not be switched and the only possibility is to switch the edges of a perfect matching $M$ in $G$ to a perfect matching among the nodes of degree 1 in $G'$.
As the resulting graph is Eulerian, it is connected, and therefore $G\setminus M$ must be a single cycle\,---\,a Hamiltonian cycle in~$G$.
\end{proof}

The proof of Theorem~\ref{thm:eulerian_flip_NP_complete} can be modified in a way that the graph $G'$ has maximum degree 4;
instead of adding $n$ paths to $v$, we can attach a tree with $n$ leaves whose inner nodes all have degree~2 or~4.

\section{The graph of 2-switches in Eulerian graphs is connected}\label{sec:connectedness}

Let us define the \emph{2-switch graph of Eulerian graphs} as the graph whose vertices are the Eulerian graphs with $n$ vertices and $m$ edges, and in which two vertices share an edge if and only if the corresponding Eulerian graphs can be obtained from each other by a 2-switch.
Theorem~\ref{thm:connected} thus states that the 2-switch graph of Eulerian graphs is connected.

\begin{definition}
A $k$-switch from a graph $G$ to $G'$ is called \emph{degree-preserving} if all vertices have the same degrees in $G$ and $G'$.
A $k$-switch is \emph{parity-preserving} if the parities of the degrees in $G$ and $G'$ are the same.
\end{definition}

We consider 2-switches on Eulerian graphs.
If $G$ is an Eulerian graph on which we perform a parity-preserving 2-switch, then the resulting graph $G'$ is Eulerian if and only if $G'$ is connected.

%Observe that the 2-switches are parity-preserving.
%We use a result by Taylor~\cite{taylor} on degree-preserving 2-switches.
Theorem~\ref{thm:connected} follows from the following two results.

\begin{theorem}[{Taylor~\cite[Theorem~3.3]{taylor}}]\label{thm:taylor}
Let $G = (V,E)$ and $H = (V,F)$ be two simple connected labeled graphs where the vertices have the same degrees. % (i.e., each vertex $v \in V$ has the same degree in $G$ as in $H$).
Then there is a sequence $(G = G_0, G_1, \dots , G_{\ell} = H)$ of graphs such that each graph $G_i$ is connected and $G_{i+1}$ can be obtained from $G_i$ by a degree-preserving 2-switch.
\end{theorem}

\begin{figure}[t]
\centering
\includegraphics{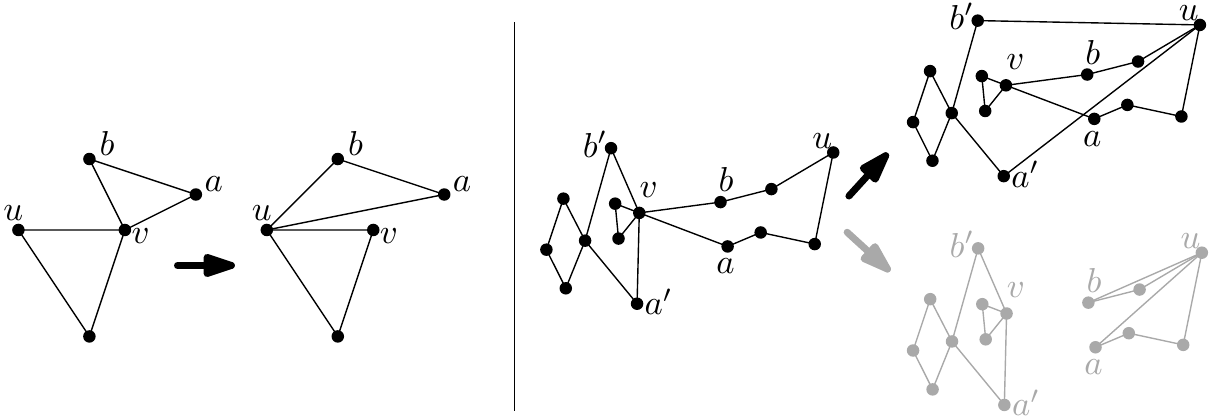}
\caption{Decreasing the degree of $v$ while increasing the degree of $u$ (left and right).
Connectivity can always be preserved (right).
The result of a disconnecting 2-switch is shown in gray.
}
\label{fig:two_increasing_switch}
\end{figure}

\begin{lemma}\label{lem:increasing}
Let $G = (V,E)$ be an Eulerian graph with two vertices $u$ and $v$ such that the degree $d_G(v)$ of $v$ in $G$ is larger than $d_G(u)$.
Then there is a parity-preserving 2-switch that increases the degree of $u$ and decreases the degree of $v$ such that the resulting graph remains connected.
\end{lemma}
\begin{proof}
We need to find two vertices $a$ and $b$ such that $av$ and $bv$ are edges of $G$, while $au$ and $bu$ are not.
Indeed, as $G$ is Eulerian, we have $d_G(u) \leq d_G(v) - 2$, so there are two vertices $a$ and $b$ in the neighborhood of $v$ that are not in the neighborhood of~$u$.
See \figurename~\ref{fig:two_increasing_switch}~(left).
If the graph stays connected after replacing $av$ and $bv$ by $au$ and $bu$, we are done, so suppose it is not connected and let $K$ be the connected component containing $v$.
See \figurename~\ref{fig:two_increasing_switch}~(right).
Note that $d_G(v) \geq 4$, so $d_K(v) \geq 2$.
Hence, there exist two edges $a'v$ and $b'v$ in $K$; $a'$ and $b'$ are not in the neighborhood of $u$ in $G$ (otherwise, the switch would not have disconnected the graph).
We perform a 2-switch on $G$ replacing $a'v$ and $b'v$ by $a'u$ and $b'u$ and argue that the resulting graph is connected.
The vertex $v$ remains in the same connected component as $u$ (there is a path via $a$ or $b$), and all other vertices of $K$ are connected to $u$ by a path via $a'$, $b'$, or~$v$.
Since also for all other vertices there is still a connection to $u$, the lemma follows.
\end{proof}

Using Lemma~\ref{lem:increasing}, we can transform any Eulerian graph $G$ on a vertex set $V = \{v_1, \dots, v_n\}$ into an Eulerian graph $G'$ such that the vertices $\{v_1, \dots, v_i\}$ have degree $d$, and the vertices $\{v_{i+1}, \dots, v_n\}$ have degree $d+2$, for some $i \leq n$ and $d$ (which only depend on $n$ and $m$).
Given two Eulerian graphs, both with $n$ vertices and $m$ edges, we can transform both with 2-switches into Eulerian graphs where the vertices have the same degrees.
Theorem~\ref{thm:taylor} tells us that we can transform one such graph into any other, and thus Theorem~\ref{thm:connected} follows.

Given Theorem~\ref{thm:connected}, it is natural to ask whether we can efficiently determine the distance of two Eulerian graphs in the 2-switch graph.
If the two graphs are (labeled) cycles, the answer is negative since this corresponds to sorting circular permutations by reversals, a problem that was shown to be NP-hard by Solomon et al.~\cite{solomon2003sorting}.

\section{Generating an Eulerian graph with given parameters}
There is a large body of work on the problem of counting
% the number of Eulerian tours in a given graph~$G$~\cite{brightwell2004note,ge2012complexity, and on counting 
the number of Eulerian graphs with a given number of vertices, see, e.g.,~\cite{read1962euler,robinson1969enumeration}.
The following related question seems to have attracted less attention:
For which combinations of number of edges and number of vertices do there exist Eulerian graphs?
Observe that any graph with vertices of even degree at least 2 can be made Eulerian by a sequence of 2-switches: pick one edge in each of two disconnected components and perform a 2-switch.
Hence, the question boils down to valid combinations of cardinalities for such graphs.
Given the desired numbers, one way would thus be to devise a corresponding degree sequence that satisfies the conditions of the Erd\H{o}s-Gallai theorem~\cite{erdos_gallai}.
As our computational experiments on 2-switch graphs required a construction of an Eulerian graph with the given elements, we explicitly identified such pairs in Theorem~\ref{thm:generatingEulerianGraphs}, whose proof is given in the appendix.

\begin{theorem}
\label{thm:generatingEulerianGraphs}
Let $P \subset \mathbb{N}^2$ be the set of all pairs of integers $(n,m)$ for which
\begin{enumerate}
\item $3\leq n \leq m \leq \binom{n}{2}$,
\item $m \not= \binom{n}{2} - i$ for $n$ odd and $i \in \{1,2\}$,
\item $m \not= \binom{n}{2} - i$ for $n$ even and $i \in \{0,1,\ldots,\frac{n}{2}-1\}$.
\end{enumerate}
There exists a simple Eulerian graph with $n$ vertices and $m$ edges if and only if $(n,m) \in P$.
\end{theorem}
It is worth to stress that our proof is constructive; that is, for a given pair $(n,m)\in P$, we provide an Eulerian $n$-vertex graph with $m$ edges.  

%\section{Conclusion}
\paragraph{Open problems.}
%We conclude with some open problems.
%\item What is the complexity of determining the minimum number of 1-swaps to obtain other classes of graphs, like chordal or bipartite graphs, similar to~\cite{edge_modification_problems}?
\begin{itemize}
\item We know that finding the 2-switch distance is NP-hard for labeled cycles (i.e., $n=m$).
%For which combinations of $n$ and $m$ is the problem hard?
Does the problem remain hard for denser graphs?
For example, it is clearly not hard for $m = \binom{n}{2}-c$ for any constant $c$.
%We could stick a complete graph of polynomial size to the cycle, obtaining something like $m = \Theta\left(\binom{n - \sqrt[c]{n}}{2}\right)$ for any constant~$c$.
Similar results exist for graphs with a fixed degree sequence~\cite{bereg}.
What about the hardness of $k$-switches for $k > 2$?
Is the problem hard for unlabeled graphs?

\item Is the 2-switch graph of Eulerian graphs Hamiltonian?
Our computer experiments showed that this is true for Eulerian graphs of up to eight vertices.
The analogous question for graphs with fixed degree sequence was raised by Brualdi~\cite{brualdi} and is apparently still open as well~\cite{barrus}.
\end{itemize}

\paragraph*{Acknowledgments}
We thank Herbert Fleischner and Patrick Schnider for valuable discussions. % on Theorem~\ref{thm:connected}.



\bibliographystyle{abbrvnat} 
\bibliography{references}

\newpage
\appendix

\section{Proof of Theorem~\ref{thm:generatingEulerianGraphs}}
Let $(n,m) \not\in P$. If $n < 3, n > m$ or if $m > \binom{n}{2}$, then clearly there is no Eulerian graph with $n$ vertices and $m$ edges.
If $n$ is odd, the only graph with $n$ vertices and $m = \binom{n}{2}-i$ edges, where $i \in \{1,2\}$, is the complete graph without 1 or 2 edges, which contains a vertex of odd degree and is therefore not Eulerian. 
If $n$ is even, then all vertices of the complete graph with $n$ vertices have an odd degree and after removing $i \in \{0,1,\ldots,\frac{n}{2}-1\}$ edges there is still some vertex of odd degree, and hence any resulting graph cannot be Eulerian.

For $(n,m) \in P$ we show by induction on $n$ that there is an Eulerian graph with $n$ vertices and $m$ edges.
As a base case let $(3,m) \in P$ which implies that $m = 3$.
The triangle is the witness graph in this case.
For $n > 3$ consider a pair $(n,m) \in P$.
If $(n-1,m-1) \in P$, then let $G = (V,E)$ be an Eulerian graph with $n-1$ vertices and $m-1$ edges and let $e = \{u,v\} \in E$.
For some $w\notin V$, the graph $(V\cup \{w\},(E\setminus \{e\}) \cup \{\{u,w\},\{v,w\}\})$ is connected, has $n$ vertices and $m$ edges, and all vertices have an even degree.

It remains to consider the case that $(n-1,m-1) \notin P$. 
We first prove the following auxiliary lemma.
\begin{lemma}
\label{lem:large-case-lemma}
If $(n,m) \in P$ and $m \geq \binom{n}{2} - n + 1$, then there is an Eulerian graph on $n$ vertices and $m$ edges.
\end{lemma}
\begin{proof}
Let $j = \binom{n}{2}-m$, which implies that $j<n$.
If $n$ is odd, then $j > 2$.
We remove a cycle of length $j$ from $K_n$.
The resulting graph is Eulerian with $n$ vertices and $m$ edges.
If $n$ is even, then $j \geq \frac{n}{2}$, as $(n,m) \in P$.
Consider the graph $K_{n-1}$ minus a matching of size $\ell = n-1-j \leq n-1-\frac{n}{2} = \frac{n}{2}-1$. Add to this graph a vertex and connect it to all vertices adjacent to the removed matching.
The resulting graph is connected, all vertices have an even degree and it has $n$ vertices and $\binom{n-1}{2} + \ell = \binom{n}{2}-j$ edges.
\end{proof}

We now consider the different cases for which $(n-1,m-1) \not\in P$:

\begin{itemize}
\item If $n-1 < 3$, then $n = 3$ which we have already covered. 
Since $(n,m)\in P$, we must have $n-1 \leq m-1$.
If $m-1 > \binom{n-1}{2}$, then $m > \binom{n}{2} - (n-2)$. 
By Lemma~\ref{lem:large-case-lemma} there is an Eulerian graph on $n$ vertices and $m$ edges.
\item If $n-1$ is odd and $m -1 = \binom{n-1}{2}-i$ for $i \in \{1,2\}$, then $m = \binom{n}{2}-(n-1)-i+1 = \binom{n}{2}-j$ where $j \in \{n,n-1\}$. 
If $j= n-1$ we are done by~Lemma~\ref{lem:large-case-lemma}.
If $j = n$, it is enough to consider $n \geq 6$ since $(n,m) \in P$.
Consider the complete graph $K_n$ and partition the vertices into two parts $V$ and $U$ of size $\frac{n}{2}$ each.
Remove from $K_n$ a perfect matching between $V$ and $U$ and a cycle of length $\frac{n}{2}$ that covers $V$.
The resulting graph has the desired number of edges and vertices, and each vertex has even degree.
It also clearly stays connected.  
\item If $n-1$ is even and $m -1 = \binom{n-1}{2}-i$ for some $i \in \{0,1,\ldots,\frac{n-1}{2}-1\}$, then $m = \binom{n}{2} - (n-2) - i = \binom{n}{2} - j$ where $j \in \{n-2,n-1,\ldots,\frac{3n-7}{2}\}$.
If $j \in \{n-2,n-1\}$, then we are done by~Lemma~\ref{lem:large-case-lemma}.
Otherwise $n \leq j \leq \frac{3n-7}{2}$.
We only need to consider odd $n \geq 5$.
There is an Eulerian graph $G$ on $n-1$ vertices and $j$ edges by the induction hypothesis.
Take the graph $K_{n-1}$, remove the edges in $G$ and add a new node connected to all vertices.
The result is a connected graph, all vertices have an even degree and it has $n$ vertices and $\binom{n-1}{2} - j + n-1$ = $\binom{n}{2} - j$ edges, as required.\qedhere
\end{itemize}

\end{document}